\documentclass[pra,superscriptaddress,twocolumn,amssymb]{revtex4}

\usepackage[utf8]{inputenc}
\usepackage{times}

\usepackage{amsmath}
\usepackage{amssymb}
\usepackage{amsthm}
\usepackage{mathtools}
\usepackage{physics}
\usepackage{bm}
\usepackage{bbm}
\usepackage{braket}
\usepackage{accents}
\usepackage{upgreek}

\usepackage[pdftex]{graphicx}
\DeclareGraphicsExtensions{.png,.pdf,.eps,.jpg}
\graphicspath{{./figs/}}
\usepackage{subfigure}
\usepackage{array}
\usepackage{multirow}
\newcolumntype{M}[1]{>{\centering\arraybackslash}m{#1}}
\newcolumntype{N}{@{}m{0pt}@{}}

\usepackage{tikz}
\usetikzlibrary{arrows.meta, positioning, decorations.pathreplacing,
	decorations.pathmorphing, calc, shapes.geometric}
\tikzstyle{box}   = [draw, fill=blue!10, text centered, rectangle,
minimum width=2cm, minimum height=1cm]
\tikzstyle{arrow} = [thick, ->, >=stealth]

\usepackage{xcolor}
\usepackage[normalem]{ulem}

\usepackage{tcolorbox}
\tcbuselibrary{breakable}
\usepackage{placeins}
\usepackage{dcolumn}

\usepackage{hyperref}
\definecolor{myurlcolor}{rgb}{0,0,0.7}
\definecolor{myrefcolor}{rgb}{0.8,0,0}
\hypersetup{
	unicode       = true,
	pdfusetitle,
	bookmarks     = false,
	breaklinks    = false,
	pdfborder     = {0 0 0},
	backref       = false,
	colorlinks    = true,
	linkcolor     = myrefcolor,
	citecolor     = myurlcolor,
	urlcolor      = myurlcolor
}

\newtheorem{theorem}{Theorem}

\newtheorem{lemma}{Lemma}

\newtheorem{definition}{Definition}

\ifx\proof\undefined
\newenvironment{proof}[1][\proofname]{\par
	\normalfont\topsep6\p@\@plus6\p@\relax
	\trivlist
	\itemindent\parindent
	\item[\hskip\labelsep\scshape #1]\ignorespaces
}{%
	\endtrivlist\@endpefalse
}
\providecommand{\proofname}{Proof}
\fi

\newcommand{\eqnref}[1]{Eq.~(\ref{#1})}

\newcommand{\secref}[1]{Sec.~\ref{#1}}
\newcommand{\appref}[1]{App.~\ref{#1}}

\renewcommand\Tr{\operatorname{Tr}}

\renewcommand{\i}{\mathrm{i}}

\newcommand{\bpm}{\begin{pmatrix}}
	\newcommand{\epm}{\end{pmatrix}}
\newcommand{\beq}{\begin{equation}}
	\newcommand{\eeq}{\end{equation}}
\newcommand{\ba}{\begin{align}}
	\newcommand{\ea}{\end{align}}
\newcommand{\bi}{\begin{itemize}}
	\newcommand{\ei}{\end{itemize}}

\newcommand{\id}{\mathrm{id}}
\newcommand{\calH}{\mathcal{H}}

\newcommand{\fru}{\mathfrak{u}}
\newcommand{\bfone}{\mathbbm{1}}

\usepackage{orcidlink}

\begin{document}
	
\title{Entanglement depth and ancilla efficiency in quantum channel estimation}
\author{Javid Naikoo\,\orcidlink{0000-0002-1825-0097}}
\email{naikooja@fzu.cz}
\affiliation{Joint Laboratory of Optics, Faculty of Science, Palack\'{y} University, Czech 	Republic, 17. listopadu 12, 779~00 Olomouc, Czech Republic}
\affiliation{Joint Laboratory of Optics of Palack\'{y} University and Institute of Physics of CAS, Institute of Physics of CAS, 17. listopadu 50a, 779 00 Olomouc, Czech Republic}

	\date{\today}
	
\begin{abstract}
We study the role of ancillary entanglement in quantum channel parameter estimation and investigate the minimal ancilla dimension required to achieve the maximum Fisher information. We introduce the $k$-ancilla Fisher information, which quantifies the optimal estimation precision achievable with
input states of rank at most $k$, and derive a variational characterization in terms of a rank-constrained optimization problem. This formulation leads to a simple characterization of the minimum ancilla dimension $k^*$ required for optimal estimation, given by the minimum rank among the maximizers of the associated variational problem. We further identify sufficient conditions under which ancillary entanglement provides no advantage, including channel families admitting a fixed measure-and-prepare representation and channels satisfying a natural horizontality condition. In addition, we derive a bound on the incremental gain in Fisher information obtained by increasing the ancilla dimension under suitable structural assumptions on the optimal input states. The general results are illustrated through explicit examples, including unitary channels, qubit depolarizing  and amplitude damping channels. These results provide a systematic framework for understanding and quantifying the entanglement resources required for optimal quantum channel estimation.
\end{abstract}
\maketitle
\section{Introduction}
Quantum parameter estimation -- the problem of inferring an unknown parameter $\theta$ from measurements on quantum states -- is the central task of quantum metrology~\cite{Helstrom1976,Holevo1982,Braunstein1994,Giovannetti2004, 	Giovannetti2006}. The problem of finding an optimal estimation scheme for a given family of quantum \emph{channels}, rather than a given
family of states, is known as the \emph{quantum channel identification 	problem}, introduced in~\cite{Fujiwara2001} and further developed in~\cite{Fujiwara2003,Fujiwara2008}. In this problem, one is given a one-parameter family of quantum channels $\{\Phi_\theta:\theta\in\Theta\}$ and aims to determine the unknown parameter $\theta$ with the highest possible precision. Unlike ordinary quantum state estimation, where the optimization is performed only over the measurement
procedure, quantum channel identification requires an additional optimization over the input state used to probe the channel. This input may be entangled with an ancillary system, making the role of entanglement a central question
in the theory of quantum channel estimation.

The role of ancilla-assisted strategies in noisy quantum channel estimation has been studied extensively, establishing both when entanglement with an external ancilla helps and how large the required ancilla must
be~\cite{Fujiwara2008,Escher2011, Demkowicz2014,Demkowicz2015, 	Huang2016}. More broadly, entanglement as a resource for quantum metrology has been investigated from both operational and resource-theoretic viewpoints, including its necessity for surpassing shot-noise scaling and its role in noiseless versus noisy settings~\cite{Toth2014,Pezze2018}.
The present work contributes to this line of research by characterizing, within the estimation setting, the entanglement resources required for optimal quantum channel identification, and by identifying structural conditions under which such resources provide no advantage.
\begin{figure}
	\includegraphics[width=\linewidth]{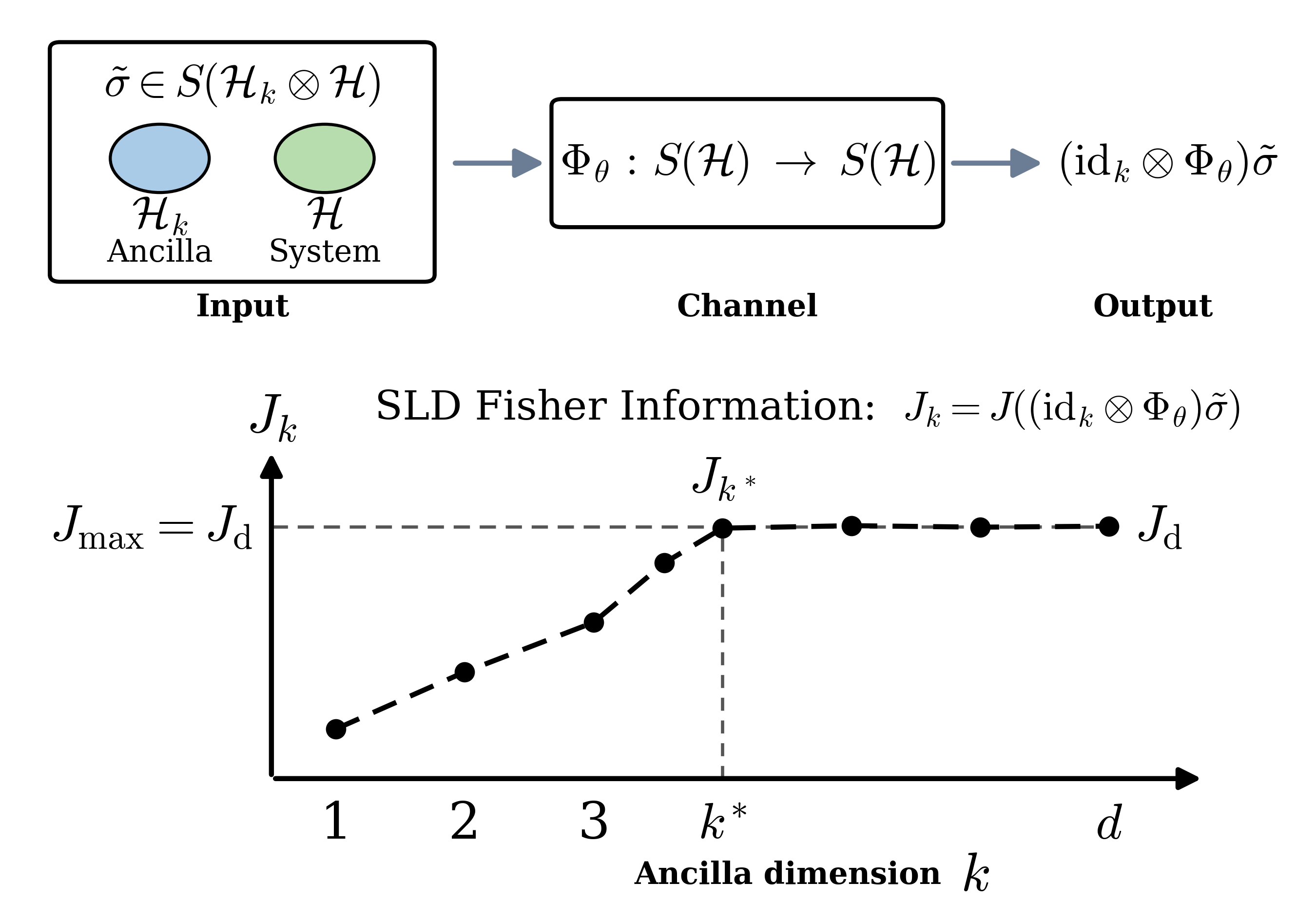}
\caption{\textbf{Schematic of the optimal ancilla dimension problem.}
	A probe state $\tilde{\sigma}\in S(\mathcal{H}_k\otimes\mathcal{H})$, prepared using a $k$-dimensional ancillary system, is sent through the parameterized quantum channel $\Phi_\theta$. The output state $(\mathrm{id}_k\otimes\Phi_\theta)(\tilde{\sigma})$ is used to estimate the parameter $\theta$, with precision quantified by the SLD Fisher information $J_k$. The plot of $J_k$ versus the ancilla dimension $k$ is illustrative only and conveys the qualitative behavior of the attainable Fisher information. A main objective of this work is to characterize the smallest ancilla dimension $k^*$ such that $J_k=J_d$, where $d=\dim\mathcal{H}$.}
	\label{fig:schematic}
\end{figure}

Let $\mathcal{H}$ be a finite-dimensional complex Hilbert space describing the
input system, and let $\mathcal{B}(\mathcal{H})$ and $S(\mathcal{H})$ denote the
sets of linear operators and density matrices on $\mathcal{H}$, respectively.
A quantum channel is represented by a completely positive trace-preserving map
$\Phi_\theta:S(\mathcal{H})\rightarrow S(\mathcal{H})$. Given an input state
$\tilde{\sigma}\in S(\mathcal{H}\otimes\mathcal{H})$, where the first copy of
$\mathcal{H}$ represents the ancillary system and the second copy represents
the input system, the output state is
\begin{equation}
	\rho_\theta=(\id\otimes\Phi_\theta)(\tilde{\sigma}).
\end{equation}
The ultimate statistical distinguishability of the parameter $\theta$ is
quantified by the symmetric logarithmic derivative (SLD) Fisher information.
Thus, the central optimization problem in quantum channel identification is the
maximization  $\max_{\tilde{\sigma}\in S(\mathcal{H}\otimes\mathcal{H})}
J\!\left((\id\otimes\Phi_\theta)(\tilde{\sigma})\right)$. This optimization problem admits an equivalent operator formulation in terms of the Kraus representation of the channel~\cite{Fujiwara2008} 
\begin{equation}
	\max_{\tilde{\sigma}\in S(\mathcal{H}\otimes\mathcal{H})}
	J\!\left((\id\otimes\Phi_\theta)(\tilde{\sigma})\right) = 4\min_{A(\theta)} \left\|
	\sum_j \dot{A}_j(\theta)^\dagger\dot{A}_j(\theta) \right\|, \label{eq:FI_main}
\end{equation}
where the minimum is taken over all smooth families of Kraus representations
$\{A_j(\theta)\}$ of $\Phi_\theta$. This result establishes that entanglement
with an ancillary system can be a useful resource for quantum channel
identification, since the optimal input is not necessarily a product state.

While Eq.~\eqref{eq:FI_main} provides a characterization of the ultimate Fisher information attainable in channel identification, it naturally motivates further questions concerning the resources required to achieve this optimum. In particular, one may ask what is the smallest ancilla dimension sufficient
for optimal estimation? Under what conditions can the same precision be achieved without ancillary entanglement? How does the optimal Fisher information scale as the allowed amount of ancillary entanglement is increased? Figure~\ref{fig:schematic} provides a schematic overview of the ancilla-assisted estimation problem and the resulting dependence of $J_k$ on the ancilla dimension $k$.

In this work, we show that these questions can be naturally formulated in terms of the rank structure of optimal input states. Building on the variational characterization of the channel Fisher information, we use the associated concave function
\begin{equation}
	g(\sigma)
	:=
	\min_{\mathrm{i}X\in\mathfrak{u}(q)}
	f_\theta(\sigma,X),
	\label{eq:g_def}
\end{equation}
where $f_\theta(\sigma,X)$ is the objective function defined in~\eqnref{eq:f_def}. We start by showing that the optimal SLD Fisher information achievable using input states of rank at most $k$ is given by
\begin{equation}
	J_k(\Phi_\theta)
	=
	\max_{\sigma\in S_k(\mathcal{H})}g(\sigma),
\end{equation}
where $S_k(\mathcal{H})$ denotes the set of states with rank at most $k$.
Consequently, the minimum entanglement dimension required to attain the
ultimate Fisher information is determined by the rank of the maximizers of
$g$. In particular, the optimal Schmidt rank is characterized by the smallest
rank among all global maximizers of this concave optimization problem.

We further investigate conditions under which ancillary entanglement is
unnecessary. We derive sufficient conditions ensuring that the extended and
unextended Fisher informations coincide,
\begin{equation}
	J_d(\Phi_\theta)=J_1(\Phi_\theta),
\end{equation}
meaning that the optimal channel identification strategy can be implemented
using a product input. These conditions include channel families admitting a
fixed measure-and-prepare representation and cases in which the relevant
generator contribution becomes linear in the input state, reducing the
optimization to an eigenvalue problem.

\section{Setup and notation}
A quantum channel admits a Kraus (operator-sum) representation $\Phi(\rho)\!=\!\sum_j A_j\rho A_j^{\dagger}$ with $\sum_j A_j^{\dagger}A_j\!=\!\bfone$~\cite{Kraus1983}. The minimum number of Kraus operators is the \emph{Kraus rank} $r\!=\!\rank\Phi$, defined as the rank of the Choi matrix $(\id\otimes\Phi)(|\Omega\rangle\langle\Omega|)$, where $|\Omega\rangle\!=\!d^{-1/2}\sum_i |ii\rangle$ is the maximally entangled state. If $\{A_j\}_{j=1}^r$ and $\{B_k\}_{k=1}^q$ are Kraus representations of the same channel, then there exists a partial isometry $Q\in\mathbb{C}^{r\times q}$ satisfying $QQ^{\dagger}\!=\! \bfone_r$ such that $B_k\!=\! \sum_j A_jQ_{jk}$. In particular, if both representations are minimal ($q\!=\! r$), then $Q$ is unitary~\cite{Fujiwara2008}.
	
For a smooth family $\{\Phi_\theta\}$ with smooth generator family $\{A_j(\theta)\}$, the function $f_\theta(\sigma, X)$ appearing in the variational formula is~\cite{Fujiwara2008}
\begin{align}
f_\theta(\sigma, X)
&:= \sum_k \Tr\, \left[\sigma\dot{B}_k^{\dagger}\dot{B}_k \right]   \nonumber\\
&\quad- \sum_{k,\ell}\Tr\!\left[\i \, \sigma (\dot{B}_k^{\dagger}B_\ell - B_k^{\dagger}\dot{B}_\ell)\right]\,x_{k\ell} \nonumber \\&\quad + \sum_{m} \bigg( \sum_{k,\ell}\Tr\, \left[\sigma B_k^{\dagger}B_\ell \right] \,x_{km} \, \overline{x_{\ell m}} \bigg), 
\label{eq:f_def}
\end{align}
where $B(\theta) \!=\! \{B_k(\theta)\}$ is a fixed smooth reference generator, $A_j(\theta) \!=\! \sum_k B_k(\theta)u_{kj}(\theta)$ for a smooth unitary family $U(\theta) \!=\! [u_{kj}(\theta)]$, and $\i X \!=\! U\dot{U}^{\dagger}\in\fru(q)$. We define the Gram matrix $G(\sigma)_{k\ell} \!:=\! \Tr[\sigma B_k^{\dagger}B_\ell]$, the velocity-overlap matrix $C(\sigma)_{\ell k} \!:=\! \Tr[\sigma(\dot{B}_k^{\dagger}B_\ell - B_k^{\dagger}\dot{B}_\ell)]$, and  $t(\sigma) \!:=\! \sum_k \Tr\, \left[\sigma\dot{B}_k^{\dagger}\dot{B}_k \right]$.

\subsection{Minimization of the generator functional}
We now determine the minimizer in the variational definition of $g(\sigma)$.
The so-called Gram matrix $G(\sigma)$ is Hermitian by construction, while the
``velocity-overlap" matrix $C(\sigma)$ is skew-Hermitia, and consequently 
$\i \, C(\sigma)$ is Hermitian. Using the definitions of $G(\sigma)$ and
\(C(\sigma)\), the functional in Eq.~\eqref{eq:f_def} can be written as
\begin{align}
f_\theta(\sigma,X) 	&= t(\sigma)	-\Tr\!\left[\i \, C(\sigma)X\right] +\Tr\!\left[XG(\sigma)X\right].
	\label{eq:f_compact}
\end{align}
The stationary point is obtained by taking the differential with respect to $X$. 
\begin{align}
	df_\theta(\sigma,X) &= \Tr\!\left[ 	\left( G(\sigma)X + XG(\sigma) - \i \, C(\sigma) \right)dX \right].
\end{align}
Since $dX$ is an arbitrary Hermitian perturbation, the stationary point $X_\ast$ satisfies
\begin{equation}
G(\sigma)X_\ast + X_\ast G(\sigma) 	= \i \, C(\sigma).
\label{eq:lyapunov}
\end{equation}
This is a continuous Lyapunov equation. Assuming $G(\sigma)$ to be positive definite, we consider the linear map
\begin{equation}
\mathcal{L}(X) = G(\sigma)X + XG(\sigma), 
\end{equation}
such that for any nonzero Hermitian matrix $X$,
\begin{align}
\langle X,\mathcal{L}(X)\rangle_F &= \Tr\!\left[ XG(\sigma)X \right] + \Tr\!\left[ X^2G(\sigma) \right] \nonumber\\
&= 2\Tr\!\left[ XG(\sigma)X \right] = 2 || X G^{1/2} ||^2 > 0.
\end{align}
Therefore, $\mathcal{L}$ is positive definite on the real vector space of Hermitian matrices and hence invertible. Equation~\eqref{eq:lyapunov} therefore has a unique Hermitian solution $X_\ast$. In particular, when $G(\sigma)$ and $C(\sigma)$ commute, Eq.~\eqref{eq:lyapunov} simplifies to
\begin{align}
	G(\sigma)X_\ast
	+
	X_\ast G(\sigma)
	&=
	2G(\sigma)X_\ast
	=
	\i \, C(\sigma),
\end{align}
which gives the explicit solution
\begin{equation}
	X_\ast
	=
	\frac{\i}{2}
	G(\sigma)^{-1}C(\sigma).
\end{equation}
Consequently, whenever $G(\sigma)$ is positive definite, the minimizer in the
definition of $g(\sigma)$ is uniquely determined as
\[
X_\ast=\mathcal{L}^{-1}(\mathrm{i}\,C(\sigma)).
\]
Thus, the evaluation of $g(\sigma)$ reduces to the solution of a continuous Lyapunov equation, for which efficient numerical algorithms are available~\cite{BartelsStewart1972,Higham2008}. Although this characterization is not required for the theoretical developments
that follow, it provides a convenient computational method for evaluating $g(\sigma)$.	This provides an explicit characterization of the generator optimization for a fixed input state. We now turn to the remaining optimization over the input state itself, leading to a variational characterization of the $k$-ancilla Fisher information.

\section{The $k$--ancilla Fisher information}
For a $k$-dimensional ancilla $\calH_k$, define
\begin{equation}
J_k(\Phi_\theta) := \max_{\tilde{\sigma}\in S(\calH_k\otimes\calH)}
		J\!\left(({\id}_k\otimes\Phi_\theta)(\tilde{\sigma})\right).
		\label{eq:Jk_def}
\end{equation}
Our first result extends Theorem~4 of~\cite{Fujiwara2008} to arbitrary ancilla dimension.
\begin{theorem}[Variational formula for $J_k$]
\label{thm:Jk}
Let $\{\Phi_\theta\}$ be a smooth, piecewise-regular one-parameter family of quantum
channels, and let $\mathcal{H}_k$ be a $k$-dimensional ancilla space for some
$k\in\{1,\dots,d\}$. Then,  for all $\theta_0\in\Theta$:
\begin{equation}
J_k(\Phi_\theta)\big|_{\theta=\theta_0} = \max_{\sigma\in S_k(\mathcal{H})} g(\sigma)\big|_{\theta=\theta_0}, \label{eq:Jk_formula}
\end{equation}
where $g(\sigma) := \min_{\i X\in\mathfrak{u}(q)} f_{\theta_0}(\sigma,X)$, 
$f_{\theta_0}(\sigma,X)$ is the function defined in \eqref{eq:f_def}, and
\begin{equation}
S_k(\mathcal{H}) := \{\sigma\in S(\mathcal{H}) : \rank\,\sigma \leq k\}.
\end{equation}
\end{theorem}
	
\begin{proof}
The proof follows in the following three steps:
		
\textit{(i) Reduction to pure inputs.} \\
The extended channel $\mathrm{id}_k\otimes\Phi_\theta$ is a linear map, so the output
state $(\mathrm{id}_k\otimes\Phi_\theta)(\tilde\sigma)$ depends linearly on the input
$\tilde\sigma$. Since the SLD Fisher information $J(\rho_\theta)$ is convex in $\rho$,
the composition $\tilde\sigma\mapsto J((\mathrm{id}_k\otimes\Phi_\theta)(\tilde\sigma))$
is convex in $\tilde\sigma$. A convex function on a compact convex set always achieves
its maximum at an extreme point of that set~\cite{Phelps2001,GeometryQuantumStates}. The extreme points of $S(\mathcal{H}_k\otimes\mathcal{H})$ are precisely the pure states $|\psi\rangle\langle\psi|$, so
\begin{align}
&\max_{\tilde{\sigma} \in S(\mathcal{H}_k \otimes \mathcal{H})}
J \big((\mathrm{id}_k\otimes\Phi_\theta)(\tilde{\sigma})\big) \nonumber \\& \qquad 
= \max_{|\psi\rangle\langle\psi|\in\partial_e S(\mathcal{H}_k\otimes\mathcal{H})}
J\big((\mathrm{id}_k\otimes\Phi_\theta)(|\psi\rangle\langle\psi|\big).
\label{eq:step_i}
\end{align}
		
\textit{ (ii) Applying the variational formula of Theorem 1 of \cite{Fujiwara2008}.}\\
Fix a pure input $|\psi\rangle\in\mathcal{H}_k\otimes\mathcal{H}$ and a smooth generator
family $\{A_j(\theta)\}_{j=1}^q$ of $\Phi_\theta$. The output state is
\begin{align}
\tilde{\rho}_\theta &:= (\mathrm{id}_k\otimes\Phi_\theta)(|\psi\rangle\langle\psi|), \nonumber \\
&= \sum_j (\bfone \otimes A_j(\theta))|\psi\rangle\langle\psi|(\bfone \otimes A_j(\theta))^*,
\end{align}
which shows that $\{(\bfone \otimes A_j(\theta))|\psi\rangle\}_{j=1}^q$ is a $\tilde\rho_\theta$-ensemble of size $q$. The unitary freedom on generators $\{A_j\}\to\{\sum_k A_k u_{kj}\}$ induces exactly the unitary freedom on this ensemble, so Theorem~1 of \cite{Fujiwara2008} applies and gives
\begin{align}
\frac{1}{4}J(\tilde{\rho}_{\theta_0})
&= \min_{A(\theta)} \,\mathrm{Tr}_{\mathcal{H}_k\otimes\mathcal{H}} \!\left[|\psi\rangle\langle\psi|
\!\left(\bfone \otimes\sum_j\dot{A}_j^*\dot{A}_j\right)\right]_{\theta=\theta_0} \nonumber\\
&= \min_{\i X\in\mathfrak{u}(q)} f_{\theta_0}(\sigma,X),
\label{eq:step_ii}
\end{align}
where $\sigma\!:=\! \mathrm{Tr}_1|\psi\rangle\langle\psi|\in S(\mathcal{H})$ is the reduced
state obtained by tracing out the ancilla $\mathcal{H}_k$, and the second equality
follows by the same algebraic reduction as in~\cite{Fujiwara2008} (expanding in terms of
the reference generator $B(\theta)$ and setting $\i \, X\!=\! U\dot{U}^*$). In
particular, the right-hand side of~\eqref{eq:step_ii} depends on the pure state
$|\psi\rangle$ \emph{only} through its reduced state $\sigma\!=\! \mathrm{Tr}_1|\psi\rangle\langle\psi|$.
		
\textit{(iii)  Identifying the range of $\mathrm{Tr}_1$ on pure states of $\mathcal{H}_k\otimes\mathcal{H}$.}
We claim that, as $|\psi\rangle$ ranges over all pure states in $\mathcal{H}_k\otimes\mathcal{H}$, the reduced state $\sigma\!=\! \mathrm{Tr}_1|\psi\rangle\langle\psi|$ ranges precisely over 	$S_k(\mathcal{H})$. Indeed, if $\psi\rangle\!=\! \sum_{i\!=\!1}^k \sqrt{\lambda_i}\, |e_i\rangle\otimes|f_i\rangle$, 	is the Schmidt decomposition of $|\psi\rangle$, then $ \mathrm{Tr}_1|\psi\rangle\langle\psi|\!=\!\sum_{i\!=\!1}^k \lambda_i|f_i\rangle\langle f_i|$, which immediately implies that $\rank(\sigma)\le k$, i.e., $\sigma\in S_k(\mathcal H)$. Conversely, let $\sigma \!\in \! S_k(\mathcal H)$ have spectral decomposition $\sigma\!=\!\sum_{i\!=\!1}^k \lambda_i|f_i\rangle\langle f_i|$,  where zero eigenvalues are included if necessary. Choosing any orthonormal basis $\{|e_i\rangle\}_{i\!=\!1}^k$ of $\mathcal H_k$, the state $|\psi\rangle\!=\! \sum_{i=1}^k \sqrt{\lambda_i}\, |e_i\rangle\otimes|f_i\rangle$ is a purification of $\sigma$, satisfying  $\mathrm{Tr}_1|\psi\rangle\langle\psi| \!=\! \sigma$. Hence, the image of the partial trace over pure states of $\mathcal H_k\otimes\mathcal H$ is exactly $S_k(\mathcal H)$.
		
Combining these facts, the map $\mathrm{Tr}_1\!:\!\partial_e S(\mathcal{H}_k\otimes\mathcal{H})
\to S(\mathcal{H})$ is surjective onto $S_k(\mathcal{H})$ and has \emph{image} exactly 		$S_k(\mathcal{H})$. Since the right-hand side of~\eqref{eq:step_ii} depends on $|\psi\rangle$ only through $\sigma = \mathrm{Tr}_1|\psi\rangle\langle\psi|$, we can replace the maximization over pure states of $\mathcal{H}_k\otimes\mathcal{H}$ with a maximization over $\sigma\in S_k(\mathcal{H})$:
\begin{align}
J_k(\Phi_\theta)\big|_{\theta_0} &\overset{\eqref{eq:step_i}}{=}
			\max_{|\psi\rangle\langle\psi|\in\partial_e S(\mathcal{H}_k\otimes\mathcal{H})}
			4\min_{\i X\in\mathfrak{u}(q)}
			f_{\theta_0}(\mathrm{Tr}_1|\psi\rangle\langle\psi|,X)
			\nonumber\\
			&\overset{\phantom{\eqref{eq:step_i}}}{=}
			\max_{\sigma\in S_k(\mathcal{H})}
			4\min_{\i X\in\mathfrak{u}(q)}
			f_{\theta_0}(\sigma,X)
			\nonumber\\
			&\overset{\phantom{\eqref{eq:step_i}}}{=}
			4\max_{\sigma\in S_k(\mathcal{H})}\,g(\sigma)\big|_{\theta_0},
\end{align}
which is~\eqref{eq:Jk_formula}.
\end{proof}

A few remarks are in order. The sets $S_k(\mathcal{H})$ form a strictly nested chain:
\begin{equation}
S_1(\mathcal{H}) \subset S_2(\mathcal{H}) \subset \cdots \subset S_d(\mathcal{H})
		= S(\mathcal{H}),
		\label{eq:nested}
\end{equation}
where the last equality holds because every $\sigma\in S(\mathcal{H})$ has rank at most
$d \!=\! \dim\mathcal{H}$. The inclusions are strict because, for each $k$, there exist states
of rank exactly $k$ (e.g.\ the equal mixture of $k$ orthonormal pure states), which
belong to $S_k(\mathcal{H})$ but not to $S_{k-1}(\mathcal{H})$.
	
Since $g(\sigma)$ is the same function in each case and is being maximized over a larger
domain at each step, the chain~\eqref{eq:nested} immediately gives the monotonicity:
\begin{equation}
		J_1(\Phi_\theta) \leq J_2(\Phi_\theta) \leq \cdots \leq J_d(\Phi_\theta).
		\label{eq:monotonicity}
\end{equation}
Each inequality $J_{k-1} \!\leq \! J_k$ holds simply because
$S_{k-1}(\mathcal{H})\!\subset \! S_k(\mathcal{H})$ -- the maximum of $g$ over a larger set
can only be equal or larger.  The inequality is strict ($J_{k-1} \!<\! J_k$) whenever
the global maximizer $\sigma^*$ of $g$ over $S(\mathcal{H})$ has rank exactly $k$ or
larger, so that it is not accessible in the smaller domain $S_{k-1}(\mathcal{H})$.
	
At the two ends of the chain we have $J_1(\Phi_\theta) \!=\! \max_{\sigma\in\partial_e
S(\mathcal{H})}g(\sigma)$ which  is the \emph{unextended} Fisher information (no ancilla,
optimized over pure input states only), and  $J_d(\Phi_\theta) \!=\! \max_{\sigma\in
S(\mathcal{H})}g(\sigma)$ recovers the \emph{fully extended} Fisher information of
Theorem~4 in~\cite{Fujiwara2008}, since $S_d(\mathcal{H}) \!=\! S(\mathcal{H})$. The
sequence~\eqref{eq:monotonicity} therefore interpolates between these two extremes,
with each term $J_k$ quantifying precisely the Fisher information achievable when the
experimenter has access to a $k$-dimensional ancilla.

This naturally raises the question of determining the smallest ancilla
dimension for which no further improvement is possible.
	
		\section{Entanglement depth}
	
	\begin{definition}[Entanglement depth]
		We define the \emph{entanglement depth} of the channel family $\{\Phi_\theta\}$  as
		\begin{equation}
			k^*(\Phi_\theta)
			:= \min\{k\in\{1,\dots,d\} : J_k(\Phi_\theta) = J_d(\Phi_\theta)\}.
		\end{equation}
		In words, $k^*$ is the smallest ancilla dimension sufficient to achieve the maximum
		Fisher information of the fully extended channel.
	\end{definition}
	
To characterize $k^*$, we first establish a structural property of the objective function $g$. This concavity property will allow us to relate the optimal ancilla dimension to the rank of an optimal input state.
\begin{lemma}[Concavity of $g$]
\label{lem:concave}
The function $g:S(\mathcal{H})\to\mathbb{R}$ defined by $g(\sigma)=\min_{X\in\mathfrak{u}(q)}f_{\theta_0}(\sigma,X)$ is concave in $\sigma$.
\end{lemma}
\begin{proof}
	For every fixed $X$, the function
	$\sigma\mapsto f_{\theta_0}(\sigma,X)$ is affine in $\sigma$, since
	each term in \eqref{eq:f_def} depends linearly on $\sigma$. Hence, for
	any $\sigma_1,\sigma_2\in S(\mathcal H)$ and $\lambda\in[0,1]$,
	\begin{align*}
		g(\lambda\sigma_1+(1-\lambda)\sigma_2)
		&=\min_X\bigl[\lambda f(\sigma_1,X)
		+(1-\lambda)f(\sigma_2,X)\bigr]\\
		&\ge
		\lambda\min_Xf(\sigma_1,X)
		+(1-\lambda)\min_Xf(\sigma_2,X)\\
		&=
		\lambda g(\sigma_1)
		+(1-\lambda)g(\sigma_2),
	\end{align*}
	which proves that $g$ is concave.
\end{proof}

Moreover, $g$ is continuous, since $f_{\theta_0}(\sigma,X)$ is continuous in $\sigma$ and the minimization is taken over the compact set $\mathfrak{u}(q)$. Since $S(\mathcal{H})$ is compact, the maximization 	problem $\max_{\sigma\in S(\mathcal{H})}g(\sigma)$ admits at least one maximizer~\cite{Bertsekas1999}. Let $\mathcal{M}$ denote the set of all	maximizers:
	\begin{equation}
		\mathcal{M}:=\arg\max_{\sigma\in S(\mathcal{H})}g(\sigma).
	\end{equation}
	
	The set $\mathcal{M}$ is convex because it is the upper level set~\cite{BoydVandenberghe2004}
	\[
	\mathcal{M}
	=
	\{\sigma\in S(\mathcal{H}):g(\sigma)\geq \max_{\tau\in S(\mathcal{H})}g(\tau)\}
	\]
	of the concave function $g$. It is also compact because it is a closed
	subset of the compact set $S(\mathcal{H})$. In general,
	$\mathcal{M}$ may contain more than one element. The following theorem shows that $k^*$ equals the
	minimum rank over all elements of $\mathcal{M}$.
	
\begin{theorem}[Rank characterization of $k^*$]
	\label{thm:kstar}
Let
\begin{equation}
\mathcal{M} := \arg\max_{\sigma\in S(\mathcal{H})} g(\sigma)
\end{equation}
be the (non-empty, convex, compact) set of all maximizers of $g$ over $S(\mathcal{H})$,	and define
\begin{equation}
\sigma^* := \arg\min_{\sigma\in\mathcal{M}} \rank\,\sigma,
\end{equation}
i.e.\ the element of $\mathcal{M}$ with the smallest possible rank. Then
\begin{equation}
k^* = \rank\,\sigma^*.
\end{equation}
\end{theorem}

\begin{proof}
	Let $r^* \!:=\! \rank\,\sigma^*$. We prove $k^* \!=\! r^*$ by establishing the two inequalities
	$k^*\!\leq \! r^*$ and $k^*\!\geq\! r^*$ separately.
	
	\textit{(i) Upper bound: $k^*\leq r^*$, equivalently $J_{r^*} = J_d$.}\\
	  By the definition of $\sigma^*$, we have $\sigma^*\in\mathcal{M}$, meaning
	\begin{equation}
		g(\sigma^*) = \max_{\sigma\in S(\mathcal{H})} g(\sigma) = \frac{J_d}{4}.
		\label{eq:sigma_star_max}
	\end{equation}
	Since $\rank\,\sigma^* = r^*$, the state $\sigma^*$ belongs to $S_{r^*}(\mathcal{H})$.
	Therefore, when we maximize $g$ over the smaller domain $S_{r^*}(\mathcal{H})$, we
	achieve at least the value $g(\sigma^*)$:
	\begin{equation}
		\max_{\sigma\in S_{r^*}(\mathcal{H})} g(\sigma)
		\geq g(\sigma^*)
		= \max_{\sigma\in S(\mathcal{H})} g(\sigma).
	\end{equation}
	But maximizing over a subset can never exceed maximizing over the whole set, so
	$\max_{S_{r^*}(\mathcal{H})} g \leq \max_{S(\mathcal{H})} g$ as well.  Together, these lead to
	\begin{equation}
		\max_{\sigma\in S_{r^*}(\mathcal{H})} g(\sigma)
		= \max_{\sigma\in S(\mathcal{H})} g(\sigma),
	\end{equation}
	which gives $J_{r^*} = 4\max_{S_{r^*}(\mathcal{H})}g = J_d$. By the definition of $k^*$
	as the \emph{smallest} $k$ with $J_k = J_d$, we conclude $k^*\leq r^*$.
	
	\textit{(ii) Lower bound: $k^*\geq r^*$, equivalently $J_k < J_d$ for all $k < r^*$:}	Fix any $k$ with $1\leq k < r^*$. We claim
	\begin{equation}
		\max_{\sigma\in S_k(\mathcal{H})} g(\sigma) < \max_{\sigma\in S(\mathcal{H})} g(\sigma).
		\label{eq:strict_ineq}
	\end{equation}
	Suppose for contradiction that \eqnref{eq:strict_ineq} fails, i.e.\ that equality holds.
	Then there exists $\hat\sigma\in S_k(\mathcal{H})$ such that
	\begin{equation}
		g(\hat\sigma) = \max_{\sigma\in S(\mathcal{H})} g(\sigma) = g(\sigma^*).
	\end{equation}
	This means $\hat\sigma\in\mathcal{M}$: it is a maximizer of $g$ over $S(\mathcal{H})$.
	But $\hat\sigma\in S_k(\mathcal{H})$ means $\rank\,\hat\sigma\leq k < r^*$, so
	$\hat\sigma$ is a maximizer of $g$ with strictly smaller rank than $\sigma^*$. This
	contradicts the definition of $\sigma^*$ as the element of $\mathcal{M}$ with minimum
	rank.
	
	Therefore \eqnref{eq:strict_ineq} holds strictly for every $k < r^*$, giving
	$J_k = 4\max_{S_k(\mathcal{H})}g < J_d$. Since $J_k < J_d$ for all $k < r^*$, the
	smallest $k$ with $J_k = J_d$ cannot be less than $r^*$, i.e.\ $k^*\geq r^*$.
	
We conclude from $(i)$ and $(ii)$ that  $r^*\leq k^*\leq r^*$, hence $k^* = r^*$.
\end{proof}
Theorem~\ref{thm:kstar} reduces the determination of the optimal ancilla
dimension to identifying the rank of a global maximizer of $g$. While this
characterization is complete, and is illustrated by examples in \secref{sec:SpecialCases}, solving the underlying optimization problem may be challenging in general. It is therefore natural to seek readily verifiable conditions under which the optimum is attained by a pure state, so that $k^*=1$. We now establish two such sufficient conditions.

	\section{When is entanglement unnecessary?}
	Given a channel family $\Phi_\theta$, let $J_d$ denote the QFI attainable with an ancilla-assisted probe and let $J_1$ denote the QFI attainable with the best unentangled input state. We refer to the difference between these two quantities as the gap in achievable Fisher information. Our interest here is in the case of the zero gap, where $J_d=J_1$ and ancillary entanglement provides no improvement in the achievable Fisher information.

	\begin{theorem}[Zero-gap conditions]
		\label{thm:zero_gap}
		The following are sufficient conditions for $k^* = 1$ (entanglement gives no advantage):
		\begin{enumerate}
			\item[\emph{(i)}] $\Phi_\theta$ admits a \emph{measure-and-prepare representation, also known as Holevo form, with a fixed measurement}, i.e., there exists a POVM $\{M_i\}_{i}$ on $\mathcal{H}$, independent of $\theta$, and a family of states $\xi_{i,\theta}\in S(\mathcal{H})$ such that~\cite{Holevo1998}
			\begin{equation}
				\Phi_\theta(X) = \sum_i \operatorname{Tr}[M_i X]\,\xi_{i,\theta},
				\label{eq:mp_rep}
			\end{equation}
			for input state $X$, and for all $\theta$  near $\theta_0$.
			\item[\emph{(ii)}] The reference generator curve $B(\theta)$ is horizontal: $\dot{B}_k^*B_\ell = B_k^*\dot{B}_\ell$ for all $k,\ell$ (equivalently, $C(\sigma) = 0$ for all $\sigma$).
		\end{enumerate}
	\end{theorem}
	
	\begin{proof}
		\emph{(i)} \emph{Proof.} We always have the trivial inequality $J_1\leq J_d$, which follows from the following argument. For any $\tau_0\in S(\mathcal{H})$, the product state $\sigma_0\otimes\tau_0\in S(\mathcal{H}\otimes\mathcal{H})$ is a valid input, and
		\begin{equation}
			J_d
			\geq J\big((\mathrm{id}\otimes\Phi_\theta)(\sigma_0\otimes\tau_0)\big)
			= J(\Phi_\theta(\tau_0)).
		\end{equation}
		Taking the maximum over $\tau_0$ gives $J_d\geq J_1$.
		
		For the reverse inequality, suppose $\Phi_\theta$ has the Holevo form \eqref{eq:mp_rep}. In addition, assume that for every $i$ there exists a $\theta$--independent state $\rho_i\in S(\mathcal{H})$ such that
		\begin{equation}
			\Phi_\theta(\rho_i)=\xi_{i,\theta}
		\end{equation}
		for all $\theta$ in a neighborhood of $\theta_0$. Equivalently, each output state $\xi_{i,\theta}$ in the measure-and-prepare representation is itself realizable as the channel output corresponding to a fixed input state.
		
		Fix any input $\tilde\sigma\in S(\mathcal{H}\otimes\mathcal{H})$, and consider the output
		\begin{equation}
			\omega_\theta := (\mathrm{id}\otimes\Phi_\theta)(\tilde\sigma).
		\end{equation}
		Since ${M_i}$ is $\theta$--independent and acts only on the second factor, \eqref{eq:mp_rep} gives
		\begin{equation}
			\omega_\theta
			= \sum_i \big(\mathrm{id}\otimes\operatorname{Tr}[M_i,\cdot,]\big)(\tilde\sigma)
			\otimes \xi_{i,\theta}
			= \sum_i p_i,\tau_i\otimes\xi_{i,\theta},
			\label{eq:sep_output_corrected}
		\end{equation}
		where, for $p_i>0$,
		\begin{equation}
			p_i := \operatorname{Tr}\big[(\bfone\otimes M_i)\tilde\sigma\big],
			\qquad
			\tau_i := \frac{1}{p_i}
			\operatorname{Tr}_2\big[(\bfone\otimes M_i)\tilde\sigma\big].
		\end{equation}
		Terms for which $p_i=0$ can simply be omitted from the sum.
		
		Crucially, because $\tilde\sigma$ and ${M_i}$ are both $\theta$--independent, so are $p_i\geq 0$ (with $\sum_i p_i=1$) and $\tau_i\in S(\mathcal{H})$. Hence all of the $\theta$--dependence of $\omega_\theta$ is confined to the states $\xi_{i,\theta}$. Thus \eqref{eq:sep_output_corrected} is a bona fide convex decomposition of the curve $\theta\mapsto\omega_\theta$ with constant weights, and the convexity of $J$ applies directly:
		\begin{equation}
			J(\omega_\theta)
			\leq \sum_i p_i,J(\tau_i\otimes\xi_{i,\theta})
			= \sum_i p_i,J(\xi_{i,\theta}),
			\label{eq:convexity_step_corrected}
		\end{equation}
		where the equality uses the fact that $\tau_i$ is independent of $\theta$, so the Fisher information of $\tau_i\otimes\xi_{i,\theta}$ comes entirely from the $\xi_{i,\theta}$ factor. Since each prepared state $\xi_{i,\theta}$ is itself the output of $\Phi_\theta$ corresponding to a fixed input state $\rho_i$, we have
	\begin{equation}
		J(\xi_{i,\theta})
		=
		J(\Phi_\theta(\rho_i))
		\leq
		\max_{\sigma\in S(\mathcal H)}
		J(\Phi_\theta(\sigma))
		=
		J_1.
		\label{eq:each_term}
	\end{equation}
Substituting \eqref{eq:each_term} into \eqref{eq:convexity_step_corrected}, we obtain
		\begin{equation}
			J(\omega_\theta)
			\leq \sum_i p_i,J_1
			=J_1.
		\end{equation}
Since this bound holds for the output of every input
		$\tilde\sigma\in S(\mathcal{H}\otimes\mathcal{H})$, taking the maximum over all inputs gives
		\begin{equation}
			J_d
			=\max_{\tilde\sigma\in S(\mathcal{H}\otimes\mathcal{H})}
			J\big((\mathrm{id}\otimes\Phi_\theta)(\tilde\sigma)\big)
			\leq J_1.
		\end{equation}
		Combining this with the trivial inequality $J_1\leq J_d$, we obtain
		\begin{equation}
			J_1\leq J_d\leq J_1
			\implies
			J_d=J_1.
		\end{equation}
		Therefore, under the fixed-measurement Holevo representation together with the condition that each prepared state $\xi_{i,\theta}$ is itself attainable as $\Phi_\theta(\rho_i)$ from a $\theta$--independent input state $\rho_i$, no entangled input strategy can outperform the optimal unentangled input strategy and hence $k^*\!=\!1$.

		\medskip
		\emph{(ii)} When $C(\sigma)=0$, the function $g$ reduces to
		\begin{equation}
			g(\sigma)
			=
			\sum_k
			\Tr[\sigma\dot{B}_k^*\dot{B}_k]
			=
			\Tr(\sigma A),
		\end{equation}
		where
		\begin{equation}
			A
			:=
			\sum_k
			\dot{B}_k^*\dot{B}_k.
		\end{equation}
		Since $A$ is Hermitian, every state $\sigma\in S(\mathcal{H})$ admits a spectral decomposition
		$\sigma=\sum_i p_i|\psi_i\rangle\langle\psi_i|$, and therefore
		\begin{equation}
			\Tr(\sigma A)
			=
			\sum_i
			p_i
			\langle\psi_i|A|\psi_i\rangle.
		\end{equation}
		Each expectation value satisfies
		$\langle\psi_i|A|\psi_i\rangle\leq\lambda_{\max}(A)$, where
		$\lambda_{\max}(A)$ denotes the largest eigenvalue of $A$. Hence,
		\begin{equation}
			\Tr(\sigma A)
			\leq
			\lambda_{\max}(A),
		\end{equation}
		with equality if $\sigma=|\psi_{\max}\rangle\langle\psi_{\max}|$, where
		$|\psi_{\max}\rangle$ is any eigenvector corresponding to
		$\lambda_{\max}(A)$. Thus, $g$ attains its maximum at a pure state, and therefore $k^*=1$.
	\end{proof}
	
	 Condition (i) is strictly stronger than requiring $\Phi_\theta$ to be entanglement-breaking at every individual $\theta$ near $\theta_0$. Entanglement-breaking guarantees, for each fixed $\theta$, \emph{some} Holevo decomposition $\Phi_\theta(X) = \sum_i \operatorname{Tr}[M_i^\theta X]\xi_{i,\theta}$, but the POVM $\{M_i^\theta\}$ may itself vary with $\theta$ in an essential way. In that case, feeding a fixed bipartite state $\tilde\sigma$ through $\mathrm{id}\otimes\Phi_\theta$ produces a separable output at each $\theta$, but the natural convex decomposition of $\omega_\theta$ has weights $p_i(\theta)$ and marginals $\tau_i(\theta)$ that vary with $\theta$; differentiating in $\theta$ then produces extra terms (associated with distinguishing the $\theta$-dependent measurement outcomes) that are not controlled by the convexity bound for $J$, which requires $\theta$-independent weights. Condition (i) rules this out by insisting the \emph{same} measurement $\{M_i\}$ works for all $\theta$ near $\theta_0$, so that the entire $\theta$-dependence is pushed into the ``preparation'' states $\xi_{i,\theta}$ --- precisely the situation the proof requires. We discuss this distinction and its implications in more detail in \appref{ap:remarks}.

\section{A Conditional Bound on the Incremental Fisher Information Gain}
Although Theorem~\ref{thm:kstar} completely characterizes the optimal ancilla dimension, it does not quantify how rapidly the Fisher information approaches its optimal value as the ancilla dimension increases. In general, obtaining
such quantitative estimates appears difficult. Nevertheless, under an additional structural assumption on the sequence of optimal input states, one can derive a simple lower bound on the incremental gain $J_{k+1}-J_k$. The following result should therefore be viewed as a conditional estimate rather than a general characterization.

	\begin{theorem}[Ancilla extension gain bound]
		\label{thm:ancilla-gain}
		Let
		\begin{equation}
			h(k):=J_k(\Phi_\theta)
			=
			\max_{\sigma\in S_k(\mathcal H)}g(\sigma),
		\end{equation}
		where $g$ is the concave function appearing in the variational
		characterization of the channel Fisher information. Suppose that for every
		$k=1,\dots,d-1$, optimal states can be chosen  such that
		\begin{equation}
			\sigma_{k+1}^*
			=
			(1-t)\sigma_k^*
			+
			t \uptau,
		\end{equation}
		for some $t \in(0,1)$ and $\uptau \in S(\mathcal H)$.
		Then the gain obtained by increasing the ancilla dimension satisfies
		\begin{equation}
			J_{k+1}(\Phi_\theta)-J_k(\Phi_\theta)
			\geq
			t 
			\left(
			g(\uptau)-J_k(\Phi_\theta)
			\right).
		\end{equation}
	\end{theorem}
	
	\begin{proof}
		Using the assumed decomposition of the optimal rank-$(k+1)$ state and the
		concavity of $g$, we have
		\begin{equation}
			\begin{aligned}
				J_{k+1}(\Phi_\theta)
				&=
				g(\sigma_{k+1}^*)\\
				&=
				g\bigl(
				(1-t)\sigma_k^*
				+
				t \uptau_k
				\bigr)\\
				&\geq
				(1-t)g(\sigma_k^*)
				+
				t g(\uptau_k).
			\end{aligned}
		\end{equation}
		Since $\sigma_k^*$ is optimal at rank $k$, i.e., $g(\sigma_k^*)=J_k(\Phi_\theta)$, therefore we have
		\begin{equation}
			\begin{aligned}
				J_{k+1}(\Phi_\theta)
				&\geq
				(1-t)J_k(\Phi_\theta)
				+
				t g(\uptau),
			\end{aligned}
		\end{equation}
		which can be rearranged as
		\begin{equation}
			J_{k+1}(\Phi_\theta)-J_k(\Phi_\theta) \geq 	t \left( g(\uptau)-J_k(\Phi_\theta) 	\right).
		\end{equation}
		
		This proves the claim.
	\end{proof}

	Theorem~\ref{thm:ancilla-gain} gives an operational interpretation of the ancilla dimension as a metrological resource. It shows that increasing the ancilla is useful when the additional state directions accessible at dimension $k+1$ provide appreciable parameter sensitivity beyond that achievable at dimension $k$. Thus, the bound quantifies the potential statistical benefit of investing in a larger reference system such that a substantial gain in $J_k(\Phi_\theta)$ translates, through the quantum Cramér--Rao bound, into improved attainable precision in estimating $\theta$.


\section{Special cases}\label{sec:SpecialCases}
We now illustrate the preceding results on several representative channel families. These examples demonstrate how the variational characterization of $J_k$ and the rank characterization of Theorem~\ref{thm:kstar} can be used in
practice to determine the optimal ancilla dimension. They also highlight that different physical noise models exhibit qualitatively different behavior: while some channels require no ancillary entanglement, others require the full
ancilla dimension to attain the maximal SLD Fisher information.

\subsection{Unitary channels}
We first consider a unitary channel family,
\begin{equation}
	\Phi_\theta(\rho)=U_\theta\rho U_\theta^\dagger ,
\end{equation}
generated by the Hermitian operator
\begin{equation}
	H_\theta=\i \, \dot U_\theta U_\theta^\dagger ,
\end{equation}
so that $\dot U_\theta=-\i \, H_\theta U_\theta$. Since the channel admits a
single Kraus operator, $B_\theta=U_\theta$, the Gram matrix appearing in
the variational formulation reduces to a scalar,
\begin{equation}
	G(\sigma)=\Tr(\sigma U_\theta^\dagger U_\theta)=1 .
\end{equation}
The Lyapunov equation for the optimal generator correction therefore gives
$X_\ast=\i \, C(\sigma)/2$, and the minimized functional becomes
\begin{equation}
	g(\sigma) = \Tr[\sigma\dot U_\theta^\dagger\dot U_\theta] + \frac{1}{4}C(\sigma)^2. \label{eq:gUnitary}
\end{equation}
Introducing the unitarily equivalent generator  $\widetilde{H}_\theta = U_\theta^\dagger H_\theta U_\theta = \i \, U_\theta^\dagger\dot U_\theta$, and using  $\dot U_\theta^\dagger\dot U_\theta =U_\theta^\dagger H_\theta^2U_\theta$ and cyclicity of the trace, the first term becomes
\begin{equation}
	\Tr[\sigma\dot U_\theta^\dagger\dot U_\theta] = \Tr[\sigma\widetilde H_\theta^2]. \label{eq:FirstTermgUnitary}
\end{equation}
The velocity-overlap matrix $C(\sigma)$ becomes
\begin{equation}
	C(\sigma)
	=
	\Tr\!\left[
	\sigma
	(\dot U_\theta^\dagger U_\theta-U_\theta^\dagger\dot U_\theta)
	\right].
\end{equation}
This simplifies by using  $U_\theta^\dagger \dot{U}_\theta\!=\! -\i \, \widetilde{H}_\theta$, to  the following
\begin{equation}
	C(\sigma) \!=\!  2 \, \i  \Tr(\sigma\widetilde H_\theta). \label{eq:SecondTermgUnitary}
\end{equation} 
Substituting Eqs.~\eqref{eq:FirstTermgUnitary} and \eqref{eq:SecondTermgUnitary} into  \eqnref{eq:gUnitary}, we obtain
\begin{equation}
	g(\sigma)
	=
	\Tr[\sigma\widetilde H_\theta^2]
	-
	\left(
	\Tr[\sigma\widetilde H_\theta]
	\right)^2
	=
	\operatorname{Var}_{\sigma}(\widetilde H_\theta).
\end{equation}
Therefore, maximizing the channel functional reduces to maximizing the variance of a Hermitian operator. The maximum variance of a Hermitian operator is determined by the extremal eigenvalues and is attained by a pure state supported on the corresponding eigenspaces. In particular, if $\ket{h_{\max}}$ and $\ket{h_{\min}}$ denote eigenvectors associated with
the largest and smallest eigenvalues of $\widetilde H_\theta$, an optimal probe state is
\begin{equation}
	\ket{\psi_{\rm opt}} = \frac{ \ket{h_{\max}} + e^{\i\, \varphi}\ket{h_{\min}} }{\sqrt{2}},
\end{equation}
where $\varphi$ is an arbitrary phase. Hence,
\begin{equation}
	\max_{\sigma\in S(\mathcal H)}g(\sigma)
	=
	\frac{1}{4}
	\left(
	\lambda_{\max}(\widetilde H_\theta)
	-
	\lambda_{\min}(\widetilde H_\theta)
	\right)^2 .
\end{equation}
Since $\widetilde H_\theta$ and $H_\theta$ are related by a unitary, they have identical spectra, and therefore
\begin{equation}
	\max_{\sigma\in S(\mathcal H)}g(\sigma)
	=
	\frac{1}{4}
	\left(
	\lambda_{\max}(H_\theta)
	-
	\lambda_{\min}(H_\theta)
	\right)^2 .
\end{equation}
Since the maximizer of $g$ is a pure state, Theorem~\ref{thm:kstar}
immediately yields $k^*=1$. Equivalently, unitary parameter estimation does not require an ancillary
system to achieve the optimal channel precision.

\subsection{Depolarizing qubit channel}
Consider the qubit depolarizing channel
\begin{equation}
	\Phi_\theta(\rho) = (1-\theta)\rho + \frac{\theta}{2} \mathbbm{1},
\end{equation}
with parameter $\theta\in(-1/3,1)$.
Writing
\begin{equation}
	\rho
	=
	\frac12
	\left(
	\mathbbm{1}
	+a\sigma_1
	+b\sigma_2
	+c\sigma_3
	\right),
\end{equation}
the variational function $g(\rho)$ turns out to be
\begin{equation}
	g(\rho)
	=
	\frac{
		3(1+\theta)-2(a^2+b^2+c^2)
	}
	{
		2(1-\theta)(1+\theta)(1+3\theta)
	}.
\end{equation}
We now determine $k^*$ by applying the rank characterization of Theorem~\ref{thm:kstar} directly to the maximization of $g$. Since
\begin{equation}
	a^2+b^2+c^2
	=
	2\operatorname{Tr}(\rho^2)-1,
\end{equation}
with arbitrary $a, b$ subject to $a^2 + b^2 + c^2 \le 1$, the function $g$ depends only on the purity of $\rho$ and is strictly decreasing as $\operatorname{Tr}(\rho^2)$ increases. Therefore $g$ is
uniquely maximized by the maximally mixed state,
\begin{equation}
	\rho^*
	=
	\frac{\mathbbm{1}}{2},
\end{equation}
which is the unique state of minimum purity. The maximizer is therefore unique and has full rank. By
Theorem~\ref{thm:kstar},
\begin{equation}
	k^*
	=
	\operatorname{rank}(\rho^*)
	=
	2.
\end{equation}
Thus the full ancilla dimension is required to attain the maximal channel
Fisher information for the qubit depolarizing channel.

\subsection{Amplitude damping channel}
To demonstrate a non-trivial noisy channel for which ancillary entanglement provides no advantage ($k^* \!=\! 1 \!<\! d \!=\! 2$), consider the single-qubit amplitude damping channel $\mathcal{E}_\theta$, which models spontaneous emission into a thermal reservoir at zero temperature. The channel map is given by $\mathcal{E}_\theta(\rho) \!=\! A_0 \rho A_0^\dagger + A_1 \rho A_1^\dagger$ with parameter $\theta \in (0,1)$, where the Kraus operators in the computational basis are
\begin{equation}
	A_0 = \begin{pmatrix} 1 & 0 \\ 0 & \sqrt{1-\theta} \end{pmatrix}, \quad 
	A_1 = \begin{pmatrix} 0 & \sqrt{\theta} \\ 0 & 0 \end{pmatrix}.
\end{equation}
For this channel, the variational function takes the form
\begin{equation}
	g(\rho)
	=
	-\frac{a^2+b^2}{16\,\theta}
	+
	\frac{1-c}{8\,\theta(1-\theta)}.
\end{equation}
For fixed $c$, $g$ is maximized by setting $a=b=0$. The remaining expression is monotonically decreasing in $c$, and hence its maximum over the Bloch ball is attained at $c=-1$. Thus the unique maximizer is the pure state $\rho^*=|1\rangle\langle1|$, with
\begin{equation}
	\operatorname{rank}(\rho^*)=1.
\end{equation}
The rank characterization of Theorem~\ref{thm:kstar} therefore gives
\begin{equation}
	k^*
	=
	\operatorname{rank}(\rho^*)
	=
	1
	<
	d=2.
\end{equation}
Hence, in contrast to the depolarizing channel, the amplitude damping channel is optimally probed without an ancilla. In other words, ancillary entanglement provides no advantage for estimating the damping parameter in this example.

Together, these examples illustrate the full range of behaviors predicted by the general theory, from channels for which ancillary entanglement is unnecessary to channels for which the maximal ancilla dimension is required.

\section{Conclusions}
We have investigated the role of ancillary entanglement as a resource for quantum channel estimation by asking a basic quantitative question, how large must the ancilla be in order to attain the maximal SLD Fisher information? To address this question, we introduced the notion of the \emph{optimal entanglement dimension} $k^*$, defined as the smallest ancilla dimension required to achieve the Fisher information of the fully extended channel. 

A main result of this work is the variational characterization of the $k$-ancilla Fisher 
information, $J_k\!=\! \max_{\sigma\in S_k(\mathcal H)}g(\sigma)$,  which reduces the optimization over bipartite probe states to a
rank-constrained optimization over density operators on the system Hilbert space. This reduction leads directly to a complete characterization of the optimal entanglement dimension given by $k^*\!=\! \rank(\sigma^*)$, 
where $\sigma^*$ is a minimum-rank maximizer of the concave function $g$. Thus, the amount of ancillary entanglement required for optimal channel estimation is determined entirely by the rank structure of an optimal input state.

We also identified situations in which ancillary entanglement provides no advantage. In particular, we proved that $k^*=1$ for channel families admitting a fixed measure-and-prepare representation and for horizontal
generator curves. In addition, under an extra structural assumption relating optimal states for successive ancilla dimensions, we obtained a lower bound on the incremental gain $J_{k+1}-J_k$, providing a quantitative estimate of the improvement obtained by enlarging the ancilla. Finally, we illustrated the general theory through several representative examples, recovering the expected behavior for unitary channels ($k^*\!=\! 1$), showing that the qubit depolarizing channel requires the full ancilla dimension ($k^*\!=\! 2$), and demonstrating that
the amplitude damping channel again satisfies $k^*\!=\! 1$.

Several open questions remain. A complete characterization of channel families for which $k^*\!=\! 1$ beyond the sufficient conditions established here would further clarify the role of ancillary entanglement in quantum estimation. It would also be interesting to extend the variational characterization to the multiparameter setting, where the Fisher information becomes matrix-valued and the optimal ancilla dimension may depend on the estimation direction. Another natural direction is the study of repeated channel use, including adaptive estimation strategies, where one expects an analogous characterization in terms  of rank-constrained optimizations over suitable multichannel probe states.

More broadly, the framework developed here recasts the problem of optimizing ancillary resources in channel estimation as a problem in finite-dimensional convex optimization. We hope that this perspective will prove useful both for the theoretical study of quantum metrology and for the design of practical channel-estimation protocols under realistic resource constraints.

	\begin{acknowledgments}
	The author is grateful to Akio Fujiwara for his valuable comments. This work was supported by the Czech Science Foundation under Project No.~25--15775S and by the Marie Skłodowska-Curie Actions--COFUND project, co-funded by the European Union (Physics for Future, Grant Agreement No.~101081515).
	\end{acknowledgments}

%

\appendix

\section{Further remarks on Theorem~\ref{thm:zero_gap}(i)}\label{ap:remarks}

The proof of Theorem~\ref{thm:zero_gap}(i) relies on the assertion that, for every input state
$\tilde{\sigma}\in S(\mathcal{H}\otimes\mathcal{H})$, the output state
\begin{equation}
	\omega_\theta
	:=
	(\operatorname{id}\otimes\Phi_\theta)(\tilde{\sigma})
\end{equation}
admits a separable decomposition of the form
\begin{equation}
	\omega_\theta
	=
	\sum_i p_i,\tau_i\otimes\xi_{i,\theta},
	\label{eq:fixed_sep}
\end{equation}
where the probabilities $p_i$ and the states $\tau_i$ are independent of the parameter $\theta$. This property can then be combined with the convexity of the SLD Fisher information to obtain
\begin{equation}
	J(\omega_\theta)
	\le
	\sum_i p_i J(\tau_i\otimes\xi_{i,\theta})
	=
	\sum_i p_i J(\xi_{i,\theta}),
\end{equation}
where the last equality follows from the fact that the first tensor factor is independent of $\theta$.

It is important to distinguish this property from the assumption that each channel $\Phi_\theta$ is entanglement-breaking. The entanglement-breaking property guarantees only that, for each fixed value of $\theta$, the output state $\omega_\theta$ is separable. In general, this implies merely that one can write
\begin{equation}
	\omega_\theta
	=
	\sum_i p_i(\theta),
	\tau_{i,\theta}\otimes\xi_{i,\theta},
	\label{eq:general_sep}
\end{equation}
where the probabilities and the local states may themselves depend on the parameter. There is, in general, no reason to expect a separable decomposition whose probabilities and first-system states are independent of $\theta$.

This distinction is important because the usual convexity inequality for the SLD Fisher information applies directly to convex combinations with parameter-independent weights. When the probabilities depend on the parameter, one must instead use the extended convexity inequality~\cite{Alipour2015}
\begin{equation}
	J\left(
	\sum_i p_i(\theta)\rho_{i,\theta}
	\right)
	\leq
	\sum_i p_i(\theta)J(\rho_{i,\theta})
	+
	J_{\mathrm{cl}} \left({p_i(\theta)}\right),
	\label{eq:extended_convexity}
\end{equation}
where $J_{\mathrm{cl}}$ denotes the classical Fisher information of the probability distribution ${p_i(\theta)}$. Applying \eqref{eq:extended_convexity} to \eqref{eq:general_sep} gives
\begin{equation}
	J(\omega_\theta)
	\leq
	\sum_i p_i(\theta)
	J(\tau_{i,\theta}\otimes\xi_{i,\theta})
	+
	J_{\mathrm{cl}}\!\left({p_i(\theta)}\right).
\end{equation}
Using additivity of the SLD Fisher information under tensor products,
\begin{equation}
	J(\tau_{i,\theta}\otimes\xi_{i,\theta})
	=
	J(\tau_{i,\theta})
	+
	J(\xi_{i,\theta}),
\end{equation}
we obtain
\begin{equation}
	J(\omega_\theta)
	\leq
	\sum_i p_i(\theta)
	\Bigl(
	J(\tau_{i,\theta})
	+
	J(\xi_{i,\theta})
	\Bigr)
	+
	J_{\mathrm{cl}}\left({p_i(\theta)}\right).
\end{equation}
Thus, without further structure, the bound contains additional contributions arising from the parameter dependence of the mixing probabilities and of the first subsystem. These contributions are absent when the probabilities $p_i$ and the states $\tau_i$ can be chosen independently of $\theta$.

A sufficient condition ensuring such a decomposition is that the family of channels admitting a common measure-and-prepare representation,
\begin{equation}
	\Phi_\theta(\rho)
	=
	\sum_i
	\operatorname{Tr}(M_i\rho),
	\sigma_{i,\theta},
\end{equation}
where the POVM ${M_i}$ is independent of $\theta$. In this case,
\begin{equation}
	(\operatorname{id}\otimes\Phi_\theta)(\tilde{\sigma})
	=
	\sum_i
	p_i,\tau_i\otimes\sigma_{i,\theta},
\end{equation}
with
\begin{equation}
	p_i = 	\Tr \left[ 	(\bfone \otimes M_i)\tilde{\sigma} \right], \qquad \tau_i = \frac{ \Tr_2 \left[ (\bfone \otimes M_i)\tilde{\sigma}
		\right] }{p_i},
\end{equation}
both independent of $\theta$ (with terms satisfying $p_i=0$ omitted). Hence, under the common fixed-measurement assumption, the decomposition \eqref{eq:fixed_sep} is indeed available, and the convexity argument used in the proof of Theorem~\ref{thm:zero_gap}(i) is justified.

It should therefore be emphasized that the fixed-measurement assumption is stronger than the statement that each $\Phi_\theta$ is entanglement-breaking. Whether the zero-gap conclusion remains valid for arbitrary parameter-dependent families of entanglement-breaking channels requires a separate argument and does not follow from the above convexity argument alone.
\end{document}